\documentclass[final,3p,times,twocolumn]{elsarticle}
\usepackage{mathpazo}
\usepackage{latexsym}
\usepackage{amssymb}
\usepackage{amsthm}
\usepackage[cmex10]{amsmath}

\usepackage{mathtools}

\usepackage{graphicx}
\usepackage{float}
\usepackage{subfig}
\usepackage[]{placeins}
\usepackage[]{algorithm}
\usepackage{algpseudocode}
\usepackage{multirow}
\usepackage[subnum]{cases}
\usepackage{tabularx}

\usepackage[]{hyperref} % required for URLs in the bibliography file
\usepackage[]{amsrefs} % without, The elsarticle class will remove the bibliography heading by declaring \let\bibsection\relax
\renewcommand{\algorithmicforall}{\textbf{for each}}

\DeclarePairedDelimiter\ceil{\lceil}{\rceil}
\DeclarePairedDelimiter\floor{\lfloor}{\rfloor}
\DeclarePairedDelimiter\abs{\lvert}{\rvert}   % |absolute value|
\newtheorem{lemma}{Lemma}
\newtheorem{theorem}{Theorem}
\newtheorem{defn}{Definition}

\makeatletter
\renewcommand\appendix{\par
  \setcounter{section}{0}%
  \setcounter{subsection}{0}%
  \setcounter{equation}{0}%
  \setcounter{table}{0}%------------ << add
  \setcounter{figure}{0}%----------- << add
  \gdef\theequation{\@Alph\c@section.\arabic{equation}}%
  \gdef\thefigure{\@Alph\c@section.\arabic{figure}}%
  \gdef\thetable{\@Alph\c@section.\arabic{table}}%
  \gdef\thesection{\Alph{section}}%
  \@addtoreset{equation}{section}%
  \@addtoreset{table}{section}%----- << add
  \@addtoreset{figure}{section}%---- << add
}
\makeatother

\begin{document}
\begin{frontmatter}

\title{Data stream fusion for accurate quantile tracking and analysis}

\author [unile] {Massimo~Cafaro\corref{cor1}}
\ead{massimo.cafaro@unisalento.it}
\cortext[cor1]{Corresponding author}
\author [unile] {Catiuscia~Melle}
\ead{catiuscia.melle@unisalento.it}
\author [unile] {Italo~Epicoco}
\ead{italo.epicoco@unisalento.it}
\author [unile] {Marco~Pulimeno}
\ead{marco.pulimeno@unisalento.it}
\address[unile]{University of Salento, Lecce, Italy}

\begin{abstract} 
\textsc{UDDSketch} is a recent algorithm for accurate tracking of quantiles in data streams, derived from the \textsc{DDSketch} algorithm. \textsc{UDDSketch} provides accuracy guarantees covering the full range of quantiles independently of the input distribution and greatly improves the accuracy with regard to \textsc{DDSketch}. In this paper we show how to compress and fuse data streams (or datasets) by using \textsc{UDDSketch} data summaries that are fused into a new summary related to the union of the streams (or datasets) processed by the input summaries whilst preserving both the error and size guarantees provided by \textsc{UDDSketch}. This property of sketches, known as mergeability, enables parallel and distributed processing. We prove that \textsc{UDDSketch} is fully mergeable and introduce a parallel version of \textsc{UDDSketch} suitable for message-passing based architectures. We formally prove its correctness and compare it to a parallel version of \textsc{DDSketch}, showing through extensive experimental results that our parallel algorithm almost always outperforms the parallel \textsc{DDSketch} algorithm with regard to the overall accuracy in determining the quantiles. 
\end{abstract}


%\tableofcontents
\end{frontmatter}

%\tableofcontents

\section{Introduction}
\label{sec:introduction}

Mergeability of data summaries is an important property \cite{Agarwal} since it allows parallel and distributed processing of datasets. In general, given two summaries on two datasets,  mergeability means that there exists an algorithm to merge the two summaries into a single summary related to the union of the two datasets, simultaneously preserving the error and size guarantees. Big volume data streams (or big data) can therefore be compressed and fused by means of a suitable, mergeable sketch data structure.

To formally define the concept of mergeability, we shall denote by $S()$ a summarization algorithm, by $D$ a dataset, by $\epsilon$ an error parameter and by $S(D, \epsilon)$ a valid summary for $D$ with error $\epsilon$ produced by $S()$. The summarization algorithm $S()$ is mergeable if there is an algorithm $\mathcal{A}$ that, given two input summaries $S(D_1, \epsilon)$ and $S(D_2, \epsilon)$, outputs a summary $S(D_1 \uplus D_2, \epsilon)$ (here $\uplus$ stands for the multiset sum operation \cite{Syropoulos}). 

Even though mergeability is a fundamental property of data summary, merging algorithms may not be necessarily simple or may be complex to formally prove correct. In particular, merging algorithms for the problems of heavy hitters and quantiles were not known until a few years ago. 

Regarding heavy hitters, Cormode and Hadjieleftheriou presented in  2009 \cite{Cormode} a survey of existing algorithms, classifying them as either counter--based or sketch--based. In the concluding remarks, Cormode and Hadjieleftheriou stated that \textquotedblleft In the distributed data case, different parts of the input are seen by different parties (different routers in a network, or different stores making sales). The problem is then to find items which are frequent over the union of all the inputs. Again due to their linearity properties, sketches can easily solve such problems. It is less clear whether one can merge together multiple counter--based summaries to obtain a summary with the same accuracy and worst--case space bounds\textquotedblright. 

The first merging algorithm for summaries obtained by running the \textsc{Misra-Gries} algorithm \cite{Misra82} (rediscovered and improved by \cite{DemaineLM02} and \cite{Karp} and also known as \textsc{Frequent})  was published in 2011 \cite{cafaro-tempesta}. One year later, \cite{Agarwal} provided a new merge algorithm for \textsc{Frequent} and \textsc{Space Saving} \cite{Metwally2006}, showing that the summaries of these algorithms are isomorphic. The same paper also provided a merging algorithms for Greenwald-Khanna quantile summaries. Later, improved merging algorithms for both \textsc{Misra-Gries} and \textsc{Space Saving} summaries were presented \cite{Cafaro-Pulimeno} \cite{Cafaro-Pulimeno-Tempesta}.

We formally prove that our \textsc{UDDSketch} \cite{Epicoco2020147604} data summary for tracking quantiles is mergeable, design and analyze a corresponding parallel algorithm and provide extensive experimental results showing the excellent scalability and accuracy achieved. This result enables parallel and distributed processing of big volume data streams (or big data), that can be compressed and fused for accurate quantile tracking and analysis.

The rest of this paper is organized as follows. We recall related work in Section \ref{related}. The merge procedure is presented in Section \ref{mergeability} and it is formally proved to be correct in Section \ref{correctness}. Experimental results are provided and discussed in Section \ref{experiments}. Finally, we draw our conclusions in Section \ref{conclusions}.

\section{Related Work}
\label{related}

\textsc{UDDSketch} is based on the \textsc{DDSketch} algorithm \cite{Masson}, and  achieves better accuracy by using a different, carefully designed collapsing procedure. Basically, \textsc{DDSketch} allows computing quantiles in a streaming setting, with accuracy defined as follows. Let $S$ be a multi-set  of size $n$ over $\mathbb{R}$ and $R(x)$ the rank of the element $x$, (the number of elements in $S$ smaller than or equal to $x$). Then, the item $x$ whose rank $R(x)$ in the sorted multi-set $S$ is $\floor{1+q(n-1)}$ (respectively $\ceil{1+q(n-1)}$) for  $0 \leq q \leq 1$ is the lower (respectively upper) $q$-quantile  item $x_q \in S$. For instance, $x_0$ and $x_1$ are respectively the minimum and maximum element of $S$, whilst $x_{0.5}$ is the median. We are now ready to define relative accuracy.

\begin{defn} Relative accuracy. $\tilde{x}_q$ is an $\alpha$-accurate $q$-quantile if
  $\abs{\tilde{x}_q - x_q} \leq \alpha x_q$ for a given $q$-quantile item
  $x_q\in S$. A sketch data structure is an $\alpha$-accurate $(q_0,q_1)$-sketch if it can output $\alpha$-accurate $q$-quantiles for ${q_0 \leq q \leq q_1}$.
 \end{defn}
 
The \textsc{DDSketch} data summary is a collection of buckets. The algorithm handles items $x \in \mathbb R_{> 0}$ and requires in input two parameters: the first one, $\alpha$, is related to the user's defined accuracy; the second one, $m$, represents the maximum number of buckets allowed. Using $\alpha$, the algorithm derives the quantity $\gamma = \frac{1+\alpha}{1-\alpha}$ which is used to define the boundaries of the $i$th bucket $B_i$. All of the values $x$ such that $\gamma^{i-1} < x \leq \gamma^i$ fall in the bucket $B_i$, with $i = \ceil{\log_\gamma{x}}$, which is just a counter variable initially set to zero. We recall here that \textsc{DDSketch} can also handle negative values by using another sketch in which an item $x \in \mathbb R_{< 0}$ is handled by inserting $-x$.

Inserting a value is done by simply incrementing the counter by one; similarly deleting a value requires decrementing by one the corresponding counter (when a counter reaches the value zero, the corresponding bucket is discarded and thrown away). Initially the summary is empty, and buckets are dynamically added as needed. It is worth noting here that bucket indexes are dynamic as well, depending just on the input value $x$ to be inserted and on the $\gamma$ value. In order to avoid that the summary grows without bounds, when the number of buckets in the summary exceeds the maximum number of $m$ buckets, a collapsing procedure is executed. The collapse is done on the first two buckets with counts greater than zero (alternatively, it can be done on the last two buckets). Let the first two buckets be respectively $B_y$ and $B_z$, with $y < z$. Collapsing works as follows: the count stored by $B_y$ is added to $B_z$, and $B_y$ is removed from the summary. Algorithm \ref{dds} presents the pseudo-code for the insertion of a value $x$ into the summary $\mathcal{S}$. 

\begin{algorithm}
  \caption{DDSketch-Insert($x, \mathcal{S}$)}
  \label{dds}
  \begin{algorithmic}
  \Require {$x \in \mathbb R_{> 0}$}
  \State $i \leftarrow \ceil{\log_\gamma{x}}$
  \If{$B_i \in \mathcal{S}$}
  	\State $B_i \leftarrow B_i + 1$
  \Else
  	\State $B_i \leftarrow 1$
  	\State $\mathcal{S} \leftarrow \mathcal{S} \cup B_i$
  \EndIf
  \If{$\abs{\mathcal{S}} > m$}
  	\State let $B_y$ and $B_z$ be the first two buckets
    \State $B_z \leftarrow B_y + B_z$
    \State $\mathcal{S} \leftarrow \mathcal{S} \smallsetminus B_y$
  \EndIf
    \end{algorithmic}
\end{algorithm}

\textsc{UDDSketch} uses a uniform collapsing procedure that provides far better accuracy with regard to \textsc{DDSketch}. In practice, we collapse all of the buckets, two by two. Given a pair of indices $(i,  i+1)$, with $i$ an odd index and $B_i \neq 0$ or $B_{i+1} \neq 0$, we create and add to the summary a new bucket with index $j = \lceil \frac{i}{2} \rceil$, with counter value equal to the sum of the $B_i$ and $B_{i+1}$ counters. The new bucket replaces the two collapsed buckets. Algorithm \ref{unif-collapse} reports the pseudocode of the uniform collapse procedure.

\begin{algorithm}
\caption{UniformCollapse($\mathcal{S}$)}
	\label{unif-collapse}
 \begin{algorithmic}
	\Require {sketch $\mathcal{S} = \lbrace B_i \rbrace_i$}
	\ForAll{ $\lbrace i: B_i > 0 \rbrace$ }
		 \State $j \leftarrow  \lceil \frac{i}{2} \rceil$
		 \State $B'_{j} \leftarrow B'_{j} + B_{i}$
	\EndFor
	\State \Return $\mathcal{S} \leftarrow \lbrace B'_i \rbrace_i$
\end{algorithmic}
\end{algorithm}

In \cite{Epicoco2020147604} we provide a theoretical bound on the accuracy achieved by the \textsc{UDDSketch} data summary.

\section{Mergeability of \textsc{UDDSketch}}
\label{mergeability}

Letting $k(n, \epsilon)$ be the maximum size of a summary $S(D, \epsilon)$ for any $D$ consisting of $n$ items, the size of the merged summary $S(D_1 \uplus D_2, \epsilon)$ is, in general, at most $k(|D_1| + |D_2|, \epsilon)$. In our case, the maximum size of the \textsc{UDDSketch} data summary $m$ is independent of $n$, being the sketch a collection of at most $m$ buckets with $m = O(1)$ (from a practical perspective, $m$ can be a small constant; as an example, $m = 500$ is already enough to provide good accuracy). Therefore, we shall denote the maximum size of our summary as $k(m, \epsilon)$. We shall show that the size of a merged summary $S(D_1 \uplus D_2, \epsilon)$ for \textsc{UDDSketch} is  still $k(m, \epsilon)$.

Our parallel \textsc{UDDSketch} algorithm is both simple and fast. Basically, the input dataset, consisting of $n$ items, is partitioned among the available $p$ processes, so that each process $p_i$ is in charge of processing either $\ceil{\frac{n}{p}}$ or $\floor{\frac{n}{p}}$ items using its own \textsc{UDDSketch} data structure $\mathcal{S}_i$. Next, all of the processes execute a parallel reduction, using as user's defined reduction operator the  Algorithm \ref{merge}, which works as follows. 

We shall denote by $\lbrace B^i_k \rbrace_k$ the set of buckets of the sketch $\mathcal{S}_i$, and by $m$ the maximum number of buckets related to the size of a sketch. The algorithm merges two input sketches $\mathcal{S}_1$ and $\mathcal{S}_2$; without loss of generality, we assume that the $\gamma$ values for the two sketches are the same (full details shall be provided in the next Section, in which we formally prove the correctness of our merge procedure). 

An \textsc{UDDSketch} data structure $\mathcal{S}_m$, which shall be returned as the merged sketch, is initialized. The merge procedure is based on the fact that given the common $\gamma$ value, each bucket interval is fixed. Therefore, in order to merge two sketches it is enough to add the counters of buckets covering the same interval. For the remaining buckets in $\mathcal{S}_1$ and $\mathcal{S}_2$ we just  create a bucket in the merged sketch with the same count. As a consequence, merging is done by scanning the buckets of $\mathcal{S}_1$ and $\mathcal{S}_2$ and considering only those buckets whose counter is greater than zero. However, the newly created $\mathcal{S}_m$ sketch may exceed the size limit. Therefore, we check if the size of $\mathcal{S}_m$ exceeds $m$ buckets and, in case, we invoke the \textsc{UDDSketch} \textsc{UniformCollapse()} procedure to enforce the constraint on the size. Finally, we return the merged sketch $\mathcal{S}_m$.

We now analyze the computational complexity of Algorithm \ref{merge}. Initializing the merged sketch $\mathcal{S}_m$ requires $O(1)$ constant time in the worst case. Scanning $\mathcal{S}_1$ and $\mathcal{S}_2$ requires in the worst case $O(m)$ time. Indeed, there are $m$ buckets in each of the input sketches, and for each one we execute $O(1)$ operations, taking into account that searching for corresponding buckets is done through an hash table. Finally, the \textsc{UniformCollapse()} operation requires at most $O(m)$ time in the worst case (again, we just need to scan at most $m$ buckets). Taking into account that $m = O(1)$, overall the worst case computational complexity of Algorithm \ref{merge} is $O(1)$.

The computational complexity of the parallel \textsc{UDDSketch} algorithm is therefore $O(\frac{n}{p} + \log{p})$ since each process $p_i$ spends $O(\frac{n}{p})$ to insert its share of the input items in its sketch, and the parallel reduction requires $O(\log{p})$ (there are $\log{p}$ steps, each one costing $O(1)$). Finally, we remark here that Algorithm \ref{merge} can also be used in a distributed setting.

\begin{algorithm}
  \caption{Merge($\mathcal{S}_1, \mathcal{S}_2$)}
  \label{merge}
  \begin{algorithmic}
  \Require {$\mathcal{S}_1 = \lbrace B^1_i \rbrace_i, \mathcal{S}_2 = \lbrace B^2_j \rbrace_j$: sketches to be merged}
  \Ensure $\mathcal{S}_m \leftarrow \lbrace B^m_k \rbrace_k$: merged sketch
  
  \State \Call{Init}{$\mathcal{S}_m$}
  \ForAll{ $\lbrace i: B^1_i > 0 \lor B^2_i > 0 \rbrace$ }
  	\State $B^m_i \leftarrow B^1_i + B^2_i$
 \EndFor
 
\If{$\mathcal{S}_m.size > m$}
	\State \Call{UniformCollapse}{$\mathcal{S}_m$}
\EndIf

\State \Return $\mathcal{S}_m$
  
\end{algorithmic}
\end{algorithm}

\section{Correctness}
\label{correctness}

In this Section we formally prove that our parallel  \textsc{UDDSketch} algorithm is correct when executed on $p$ processors (or cores). We need the following definition.

\begin{defn}
\label{multiset}
A multiset $\mathcal{N}=(N, f)$ is a pair where $N$ is some set, called the underlying set of $\mathcal{N}$, and $f: N \rightarrow \mathbb{N}$ is a function. The generalized indicator function of $\mathcal{N}$ is 

\begin{equation}
\label{eq01}
I_\mathcal{N} (x) := \left\{ {\begin{array}{*{20}c}
   f(x) & {x \in N} , \\
   0 & {x \notin N},  \\
 \end{array} }\right.
 \end{equation}
 
\noindent where the integer--valued function $f$, for each $x \in N$, provides its \textit{multiplicity}, i.e., the number of occurrences of $x$ in $\mathcal{N}$. The cardinality of $\mathcal{N}$ is expressed by

\begin{equation}
\label{eq02}
\left\vert{\mathcal{N}}\right\vert := Card(\mathcal{N}) = \sum\limits_{x \in N} {I_\mathcal{N} (x)},
\end{equation}

\noindent whilst the cardinality of the underlying set $N$ is

\begin{equation}
\label{eq03}
\left\vert{N}\right\vert := Card(N) = \sum\limits_{x \in N} {1}.
\end{equation}
\end{defn}

\noindent A multiset (also called a \emph{bag}) essentially is a set where the duplication of elements is allowed. We also need the definition of the \textit{sum} operation \cite{Syropoulos} for multisets.

\begin{defn}
Let $\mathcal{A} = (A, f)$ and $\mathcal{B} = (B, g)$ be two multisets. The sum of $\mathcal{A}$ and $\mathcal{B}$ is the multiset whose underlying set is the union of the underlying sets and whose multiplicity function is the sum of the multiplicity functions: $\mathcal{A} \uplus \mathcal{B} = ((A \cup B), f + g)$.
\end{defn}

In the sequel, $\mathcal{N}$ will play the role of a finite
input dataset, containing $n$ items. We partition the original dataset $\mathcal{N}$, considered as a multiset, in $p$ datasets $\mathcal{N}_i$ $(i=0,\ldots,p-1)$, namely $\mathcal{N}=\biguplus_i \mathcal{N}_i$. 
Let the dataset $\mathcal{N}_i$ be assigned to the processor $p_i$, whose rank is denoted by $id$, with $id=0,\ldots, p-1$. Let also $\left\vert{\mathcal{N}_i}\right\vert$ denote the cardinality of $\mathcal{N}_i$, with $\sum_i \left\vert{\mathcal{N}_i}\right\vert=\left\vert{\mathcal{N}}\right\vert=n $.

The first step of the algorithm consists in the execution of the sequential 	\textsc{UDDSketch} algorithm (which has already been proved to be correct) on the dataset assigned to each processor $p_{i}$. Therefore, in order to prove the overall correctness of the algorithm, we just need to demonstrate that the parallel reduction is correct.

Our strategy is to prove that if a single sub-step of the parallel reduction is correct (i.e., Algorithm \ref{merge}), then we can naturally extend the proof to the $O(\log~p)$ steps of the whole parallel reduction. We begin by proving the following Lemma, which states that UDDSketch is permutation invariant with regard to insertion-only streams.

\begin{lemma}
\label{invariance}
    UDDSketch is permutation invariant with regard to insertion-only streams, i.e., it produces the same sketch regardless of the order in which the input items are inserted.
\end{lemma}

\begin{proof}
    Let $\mathcal{D} = (\Delta, \mu)$ be a multiset representing an insertion-only input stream (i.e., deleting an item is not allowed). $\Delta \subset \mathbb{R^+}$ is the underlying set of $\mathcal{D}$ and $\mu: \Delta \rightarrow \mathbb{N_{\geq 1}}$ is its multiplicity function. Let $i_{\gamma}: \Delta \rightarrow \mathbb{Z}: i_{\gamma}(x) =  \lceil \log_{\gamma} x \rceil$ denote the function which maps each item $x \in \Delta$ to the corresponding bucket in the sketch built by \textsc{UDDSketch} processing $\mathcal{D}$ and assume that the sketch can grow unbounded. Then $i_\gamma(\Delta)$, the image of $\Delta$ through the mapping function $i_\gamma$, corresponds to the set of bucket keys in the sketch summarizing the multiset $\mathcal{D}$ with a guaranteed accuracy of $\alpha = \frac{\gamma - 1}{\gamma + 1}$ and $\left|i_\gamma(\Delta)\right|$ is the number of such buckets, i.e., the size of the sketch. 
    
    Moreover, for each bucket key $k \in i_\gamma(\Delta)$, the preimage of $k$ under $i_\gamma$, denoted by $i^{-1}_\gamma(k)$, is the set of items assigned to the bucket $B_k$, and we can compute the value of a bucket $B_k$ as the sum of the multiplicities of its items in the input dataset, i.e., $B_k = \sum_{x \in i^{-1}_\gamma(k)} \mu(x)$. 
    
    Therefore, the sketch computed by UDDSketch on a dataset $\mathcal{D}$ is completely determined by the sets $i_\gamma(\Delta)$ and $i^{-1}_\gamma(k) \forall k \in i_\gamma(\Delta)$ which do not depend on the order in which the items in $\mathcal{D}$ are processed. We can represent the sketch produced by UDDSketch for the dataset $\mathcal{D}$ as the multiset $\mathcal{S} = (i_\gamma(\Delta), \beta)$, where $\beta: \Sigma \rightarrow \mathbb{N_{\geq 1}}: \beta(k) = \sum_{x \in i^{-1}_\gamma(k)} \mu(x)$.

    When the sketch is allowed to grow unbounded, the value of $\gamma$ and consequently the accuracy of the sketch is not constrained; it can be set arbitrarily and is not modified by UDDSketch. On the contrary, when a limit to the number of buckets is imposed, UDDSketch must determine the value of $\gamma$ that allows respecting that limit, i.e., the value of $\gamma$ also becomes an output of the algorithm. 

    In fact, the collapsing procedure of UDDSketch is equivalent to a change of the value of $\gamma$, which is squared in each collapse operation, and a sketch reconstruction through the  mapping function using the new $\gamma$ value. When a limit of $m$ buckets is imposed to the size of the sketch and that limit is exceeded with the current value of $\gamma$, UUDSketch squares that value and reconstructs the sketch until the constraint $\left|i_\gamma(\Delta)\right| \leq m$ is satisfied. 
    
    The characterization of the sketch as the multiset $(i_\gamma(\Delta), \beta)$ continues to hold even if collapsing operations are executed with $\gamma$ set to the value needed to respect the sketch size constraint, and the sketch remains invariant with regard to the order in which the items are processed or the order in which the collapsing operations are executed, thus proving that UDDSketch is permutation invariant when processing insertion-only streams.
\end{proof}

We consider now a single step of the parallel reduction, i.e., the case when the input dataset, represented by a multiset $\mathcal{D}$, is partitioned into the multisets $\mathcal{D}_1$ and $\mathcal{D}_2$, so that $\mathcal{D} = \mathcal{D}_1 \biguplus \mathcal{D}_2$, where $\biguplus$ represents the \textit{sum} operation \cite{Syropoulos}. We independently process  $\mathcal{D}_1$ and $\mathcal{D}_2$ with two instances of UDDSketch initialized with the same initial value of the parameter $\gamma$ and the same limit to the number of buckets. 
%Let denote by $S'_\gamma$ and $S''_\gamma$ the two sketches corresponding respectively to $\mathcal{D}_1$ and $\mathcal{D}_2$. 
 
Without loss of generality, we assume that the final values of $\gamma$ for the two sketches are the same. In fact, we prove here that this is not restrictive. Setting the same initial conditions, the sequence of values that $\gamma$ can assume due to collapses of the two sketches is the same, i.e., $\gamma \in \{\gamma_0, \gamma^2_0, \gamma^4_0, \gamma^8_0 \ldots\}$ holds for both the sketches. If the final values of $\gamma$ do not match, we can always repeatedly collapse the sketch with smaller $\gamma$ until it matches the $\gamma$ of the other sketch. We shall show that the following Theorem holds.

\begin{theorem}
\label{merge-correctness}
    Let $\mathcal{D}_1 = (\Delta_1, \mu_1)$ and $\mathcal{D}_2 = (\Delta_2, \mu_2)$ be two multisets and $\mathcal{S}_1$ and $\mathcal{S}_2$ the sketches produced by UDDSketch respectively processing $\mathcal{D}_1$ and $\mathcal{D}_2$ with a limit to the number of buckets, $m$, and an initial value of $\gamma = \gamma_0$. Denote by $\mathcal{S}_m$ the sketch obtained by merging $\mathcal{S}_1$ and $\mathcal{S}_2$ on the basis of the UDDSketch merge procedure and denote by $\mathcal{S}_g$ the sketch that UDDSketch would produce on the multiset $\mathcal{D} = (\Delta, \mu) = \mathcal{D}_1 \biguplus \mathcal{D}_2$ with the same size limit $m$ and the same initial value of $\gamma = \gamma_0$. Then, $\mathcal{S}_g = \mathcal{S}_m$.  
\end{theorem}

\begin{proof}
    We shall prove that separately computing $\mathcal{S}_1$ and $\mathcal{S}_2$ and then merging them in order to obtain $\mathcal{S}_m$, results in the same sequence of operations related to sequentially processing through UDDSketch all of the items in $\mathcal{D}$, but in a particular order. 

    Without loss of generality, we assume that the final value of $\gamma$ for $\mathcal{S}_1$ is larger than that for $\mathcal{S}_2$, the other case being symmetric. 
    
    To make it possible merging $\mathcal{S}_1$ and $\mathcal{S}_2$, we need to repeatedly collapse $\mathcal{S}_2$ until its $\gamma$ value (and consequently its mapping function) matches the one of $\mathcal{S}_1$. After this preliminary operation, all of the items available both in $\mathcal{D}_1$ and $\mathcal{D}_2$ turn out to be processed by the same mapping function although by two separate sketches. This also means that buckets with the same key in the two sketches have the same boundaries. 
    
    Denote by $\mathcal{T}$ the sketch computed by sequentially processing $\mathcal{D}$. We start the sequential procedure by first inserting in $\mathcal{T}$ all of the items in $\mathcal{D}_1$. Therefore, at the end, it holds that $\mathcal{T} = \mathcal{S}_1$. Then, we continue to insert in $\mathcal{T}$ all of the items in $\mathcal{D}_2$ that fall in buckets already present in $\mathcal{T}$. This produces the same result that we obtain in the merging procedure, when we set $\mathcal{S}_m = \mathcal{S}_1$ and increment the count of each bucket in $\mathcal{S}_m$ with the count of the bucket with the same key in the sketch $\mathcal{S}_2$, if it exists. Now, we continue to insert in $\mathcal{T}$ all of the remaining items of $\mathcal{D}_2$, which leads to the creation of new buckets in $\mathcal{T}$, but we do not collapse the sketch for now. This corresponds to adding to the sketch $\mathcal{S}_m$ all of the buckets in $\mathcal{S}_2$ with keys that are not yet in $\mathcal{S}_m$ and this concludes the first step of the merging procedure. Up to this point, consisting of the same operations, the sequential procedure on $\mathcal{D}$ and the merging procedure on $\mathcal{S}_1$ and $\mathcal{S}_2$ produce two identical sketches, $\mathcal{T} = \mathcal{S}_m$. The second step of the merging procedure consists of collapsing $\mathcal{S}_m$ until the constraint on the sketch size, $m$, is satisfied, but this constraint also holds for $\mathcal{T}$, which is  subject to the same number of collapses. Thus, the equality is maintained.

    $\mathcal{T}$ is the sketch that we obtain processing through UDDSketch the dataset $\mathcal{D}$ in a particular order of insertions and collapses, but we know from Lemma \ref{invariance} that the order of insertions and collapses is not relevant, therefore we can conclude that $T = \mathcal{S}_g$ which finally proves the thesis $\mathcal{S}_m = \mathcal{S}_g$. 

    % <primo tentativo>

    % First, we show that the theorem holds when the sketches can grow without limits and no collapse is needed. 
     
    % \noindent Under this assumption and when the input stream only includes insertions, we know from the proof of Lemma 1 that, if $i_\gamma$ is a function from $\mathbb{R}^+$ to $\mathbb{R}$, defined as $i_\gamma(x) = \lceil \log_{\gamma} x \rceil$, then $\mathcal{S}_1 = (i_\gamma(\Delta_1), \beta_1)$, $\mathcal{S}_2 = (i_\gamma(\Delta_2), \beta_2)$, and, based on how the merge procedure works when no collapse is needed, the sketch resulting from merging $\mathcal{S}_1$ and $\mathcal{S}_2$ is $\mathcal{S}_m = (i_\gamma(\Delta_1) \bigcup i_\gamma(\Delta_2), \beta_1 + \beta_2$).
    % We also know that:

    % $$\beta_1(k) = \sum_{x \in i^{-1}_\gamma(k) \bigcap \Delta_1} (\mu_1(x))$$

    % \noindent and

    % $$\beta_2(k) = \sum_{x \in i^{-1}_\gamma(k) \bigcap \Delta_2} (\mu_2(x))$$

    % \noindent and

    % $$\beta_1(k) + \beta_2(k) = \sum_{x \in i^{-1}_\gamma(k) \bigcap (\Delta_1 \bigcup \Delta_2)} (\mu_1(x) + \mu_2(x)) = \sum_{x \in i^{-1}_\gamma(k) \bigcap \Delta} (\mu(x))$$

    % \noindent Thus, we can write 
    
    % \begin{equation*}
    % \mathcal{S}_m = (i_\gamma(\Delta_1) \bigcup i_\gamma(\Delta_2), \sum_{x \in i^{-1}_\gamma(k) \bigcap \Delta} \mu(x))) = (i_\gamma(\Delta_1 \bigcup \Delta_2), \beta) = \mathcal{S}_g
    % \end{equation*}
    
    % \noindent where $\beta(k) = \sum_{x \in i^{-1}_\gamma(k) \bigcap \Delta} (\mu(x))$.

    % This proves the equality of the two sketches, $\mathcal{S}_m$ and $\mathcal{S}_g$, if we let them grow unbounded.

    % ...
\end{proof}

Lemma \ref{invariance} and Theorem \ref{merge-correctness} hold for insertion-only input streams. When the input stream also includes deletions, the permutation invariance of UDDSketch and consequently the  equality between the two sketches $\mathcal{S}_m$ and $\mathcal{S}_g$ can not be guaranteed. Anyway, the following Theorem, holds even when deletions are allowed.

\begin{theorem}
\label{final-bound}
    Let $\sigma_1$ and $\sigma_2$ be two streams including insertions and deletions of items drawn from the universe set $U = [x_{min}, x_{max}] \subset \mathbb{R}^+$  and $\mathcal{S}_1$ and $\mathcal{S}_2$ be the sketches produced by UDDSketch processing respectively $\sigma_1$ and $\sigma_2$ with the sketch size limited to $m$ buckets, and an initial value of $\gamma = \gamma_0$. Denote by $\mathcal{S}_m$ the sketch obtained by merging $\mathcal{S}_1$ and $\mathcal{S}_2$ on the basis of the UDDSketch merge procedure and denote by $\mathcal{S}_g$ the sketch that UDDSketch would produce on the stream $\sigma = \sigma_1 \biguplus \sigma_2$ with the sketch size limited to the same number of buckets, $m$, and the same initial value of $\gamma = \gamma_0$. Then, $\mathcal{S}_g$ and $\mathcal{S}_m$ have the same error bound.
\end{theorem}

\begin{proof}   
    The value of $\gamma$ during the execution of UDDSketch can only grow due to the collapses of the sketch and its final value depends on the order in which deletions are interleaved with insertions.
    The worst case scenario, when $\gamma$ reaches its largest value, happens when all of the deletions are postponed after all of the insertions. This particular order of insertions and deletions, in turn, produces a sketch with the same final value of $\gamma$ that one would obtain by processing only the insertions of the input stream and completely ignoring the deletions. In fact, deletions may change the bucket counters' values in a sketch, but not its $\gamma$ value. 

    On the other hand, an insertion-only stream falls in the hypothesis of Theorem \ref{merge-correctness}. Thus, if we consider only the insertions in $\sigma_1$, $\sigma_2$ and their concatenation $\sigma$, and ignore deletions, the two sketches $\mathcal{S}_m$ and $\mathcal{S}_g$ would be equal and have the same final value of $\gamma$. Denote by $\tilde{\gamma}$ this value, then, with regard to the original input streams with deletions, $\tilde{\gamma}$ is an upper bound on the values of $\gamma$ both for $\mathcal{S}_m$ and $\mathcal{S}_g$. 
    
    We know that the value of $\gamma$ for $\mathcal{S}_g$ is guaranteed as bounded by Theorem 3 of \cite{Epicoco2020147604}, i.e., $\gamma \leq \tilde{\gamma} \leq \left(\sqrt[m]{\frac{x_{max}}{x_{min}}}\right)^2$. Therefore, the guarantee on the accuracy of UDDSketch stated by Theorem 3 of \cite{Epicoco2020147604} continues to hold also for a sketch computed through the merge procedure.

    % <primo tentativo>

    % An upper bound on the value of $\gamma$ is guaranteed by Theorem ... which states that a sketch with at most $m$ buckets is at least an $\hat{\alpha}$-accurate sketch with $\hat{\alpha} = $.
    % The value of $\tilde{\gamma}$ depends on the values of 
    % $x_{min}$, $x_{max}$  and $m$ which are the same for $\sigma_1$, $\sigma_2$ and $\sigma$. As a consequence, $\mathcal{S}_1$, $\mathcal{S}_2$ and $\mathcal{S}_g$ have the same guarantees and share the same bound on the accuracy. What happens to $\mathcal{S}_m$? \dots
\end{proof}

\begin{lemma}
\label{parallel-reduction}
    The parallel reduction in which the sketches $\mathcal{S}_1,\cdots, \mathcal{S}_p$ are processed on $p$ processors or cores of execution is correct.
\end{lemma}

\begin{proof}
Consider a single step of the reduction, in which two sketches $\mathcal{S}_i$ and $\mathcal{S}_j$ are merged producing the sketch $\mathcal{S}_m$. By Theorem \ref{merge-correctness} and \ref{final-bound}, the sketch $\mathcal{S}_m$ is correct and subject to the same error bound of both $\mathcal{S}_i$ and $\mathcal{S}_j$. Now consider the whole reduction operation. Let $\mathcal{A} = (A, f_1)$, $\mathcal{B} = (B, f_2)$ and  $\mathcal{C} = (C, f_3)$. It can be easily shown, by reduction to the analogous properties holding in the ring of the integers, that the multiset sum operation has the following properties.

\begin{enumerate}
  \item Commutativity: $\mathcal{A} \biguplus \mathcal{B}=\mathcal{B} \biguplus \mathcal{A}$. 
  \item Associativity: $(\mathcal{A} \biguplus \mathcal{B}) \biguplus \mathcal{C}=\mathcal{A} \biguplus (\mathcal{B} \biguplus \mathcal{C})$;
  \item There exists a multiset, the null multiset $\epsilon = (\varnothing, g: x \rightarrow 0)$, such that $\mathcal{A} \biguplus \epsilon = \mathcal{A}$.
\end{enumerate}

Regarding commutativity, 

\begin{equation}
  \begin{aligned} \mathcal{A} \biguplus \mathcal{B} &=\left((A \cup B), f_{1}+f_{2}\right) \\ &=\left((A \cup B), f_{2}+f_{1}\right) \\ &=\mathcal{B} \biguplus \mathcal{A} \end{aligned}
\end{equation}

For associativity, 

\begin{equation}
  \begin{aligned} \mathcal{A} \biguplus \left(\mathcal{B}  \biguplus \mathcal{C}\right) &=\left((A \cup B \cup C), f_{1}+\left(f_{2}+f_{3}\right)\right) \\ &=\left((A \cup B \cup C),\left(f_{1}+f_{2}\right)+f_{3}\right) \\ &=\left(\mathcal{A} \biguplus \mathcal{B}\right) \biguplus \mathcal{C} \end{aligned}
\end{equation}

Finally, let $\epsilon = (\varnothing, g: x \rightarrow 0)$ be the empty multiset, i.e. the unique multiset with an empty underlying set; thus $Card(\epsilon) = 0$.

\begin{equation}
	\begin{aligned}
		\mathcal{A} \biguplus \epsilon &= ((A \cup \varnothing), f + g)\\
		&= (A, f)\\
		&= \mathcal{A}
	\end{aligned}
\end{equation}

Therefore, the merge procedure described for two multisets can be used, being associative, as a parallel reduction operator. Moreover, being also commutative, the order of evaluation must not be necessarily fixed (e.g., for non commutative user's defined operators in MPI is defined to be in ascending, process rank order, beginning with process zero) but can be changed, taking advantage of commutativity and associativity. Moreover, the final sketch obtained by the parallel reduction operator is also subject to the same error bound of the input sketches.
 
\end{proof}

\section{Experimental Results}
\label{experiments}

In this Section, we present and discuss the results of the experiments carried out for both \textsc{UDDSketch} and \textsc{DDSketch}. The aim is twofold: i) we aim at showing that the accuracy does not decrease when executing the algorithm in parallel; ii) the running time of \textsc{UDDSketch} is similar to that of \textsc{DDSketch}. 

Both algorithms have been implemented in C++. The tests have been executed on two supercomputers: Marconi100 (at CINECA, Italy) and Zeus (at Euro Mediterranean Center on Climate Change, Foundation, Italy). Marconi100 is made of 980 computing nodes equipped with 2 16-cores IBM Power9 processors, 256 GB of main memory and Mellanox Infiniband EDR DragonFly+; the code has been compiled with the PGI compiler pgc++ version 20.9-0 with optimization level O3. Zeus is a parallel cluster made of 384 computing node, each one equipped with 2 18-cores Intel Xeon Gold processors, 96 GB of main memory and Mellanox Infiniband EDR network; the code has been compiled with the Intel compiler icpc v19.0.5 with optimization level O3. The source code is freely available for inspection and reproducibility of results\footnote{https://github.com/cafaro/PUDDSKETCH}. 

\begin{table}
	\small
	\caption{Synthetic datasets}
	\label{synthetic-data}
	\centering
	%\ra{1.4}
	\begin{tabular}{@{}llll@{}}
		\textbf{Dataset} & \textbf{Min value} & \textbf{Max value} & \textbf{Distribution}\\
		\hline
		beta & $3.04 \times 10^{-2}$ & $0.99$ & $\textit{Beta}(5, 1.5)$ \\
		exponential & $1.19\times10^{-7}$ & $34.9$ & $\textit{Exp}(3.5)$ \\
		lognormal & $1.08\times10^{-3}$ & $7.91\times10^{3}$ & $\textit{Lognormal}(1, 1.5)$ \\
		normal & $39.7$ & $60.5$ & $\textit{N}(10^6, 20000)$ \\
		uniform & $2.18\times10^{-3}$ & $2.49\times10^4$ & $\textit{Unif}(5, 10^6)$ \\
	\end{tabular}
\end{table}

The tests have been performed on $5$ synthetic datasets, whose properties are summarized in Table \ref{synthetic-data}. The experiments have been executed varying the number of parallel processes and measuring the execution time, the $q_0$-accuracy, the final value of $\alpha$ and the total number of collapses for both algorithms. We recall that for \textsc{UDDSketch} the $q_0$-accuracy is equal to $0$ by construction, and for \textsc{DDSketch} the final value of $\alpha$ is equal to its initial value. The stream length and the sketch size have been kept constant for every experiments as reported in Table \ref{parameters}. The results obtained on both parallel computers are totally equivalent and showed the same behaviours; for this reason, we report here only the results on Marconi100.

\begin{table}
	\small
	\caption{Experiments parameters}
	\label{parameters}
	\centering
	%\ra{1.2}
	%\small
	\begin{tabular}{@{}ll@{}}
		\textbf{Parameter} & \textbf{Set of values}\\
		\hline
		Number of procs. (M100)& $\{32, 64, 128, 256, 512\}$ \\
		Number of procs. (Zeus)& $\{36, 72, 144, 288, 576\}$ \\
		Stream Lenght (M100) & $16\cdot10^9$ \\
		Stream Lenght (Zeus) & $18\cdot10^9$ \\
		User $\alpha$ & $0.001$ \\
		Number of buckets & $512$ \\
	\end{tabular}
\end{table}

Fig. \ref{collapses} reports the total number of collapses for \textsc{DDSketch} and \textsc{UDDSketch}. As expected, we have that \textsc{DDSketch} performs a number of collapses which is about three order of magnitude greater than those performed by \textsc{UDDSketch}. Even though the running time for both a \textsc{DDSketch} and a \textsc{UDDSketch} collapse is $O(1)$, the asymptotic notation hides a bigger constant in the case of \textsc{UDDSketch}.

\begin{figure*}[h]
	\centering
	\begin{tabular}{cc}		
		\subfloat[]{
		\includegraphics[width=0.45\textwidth]{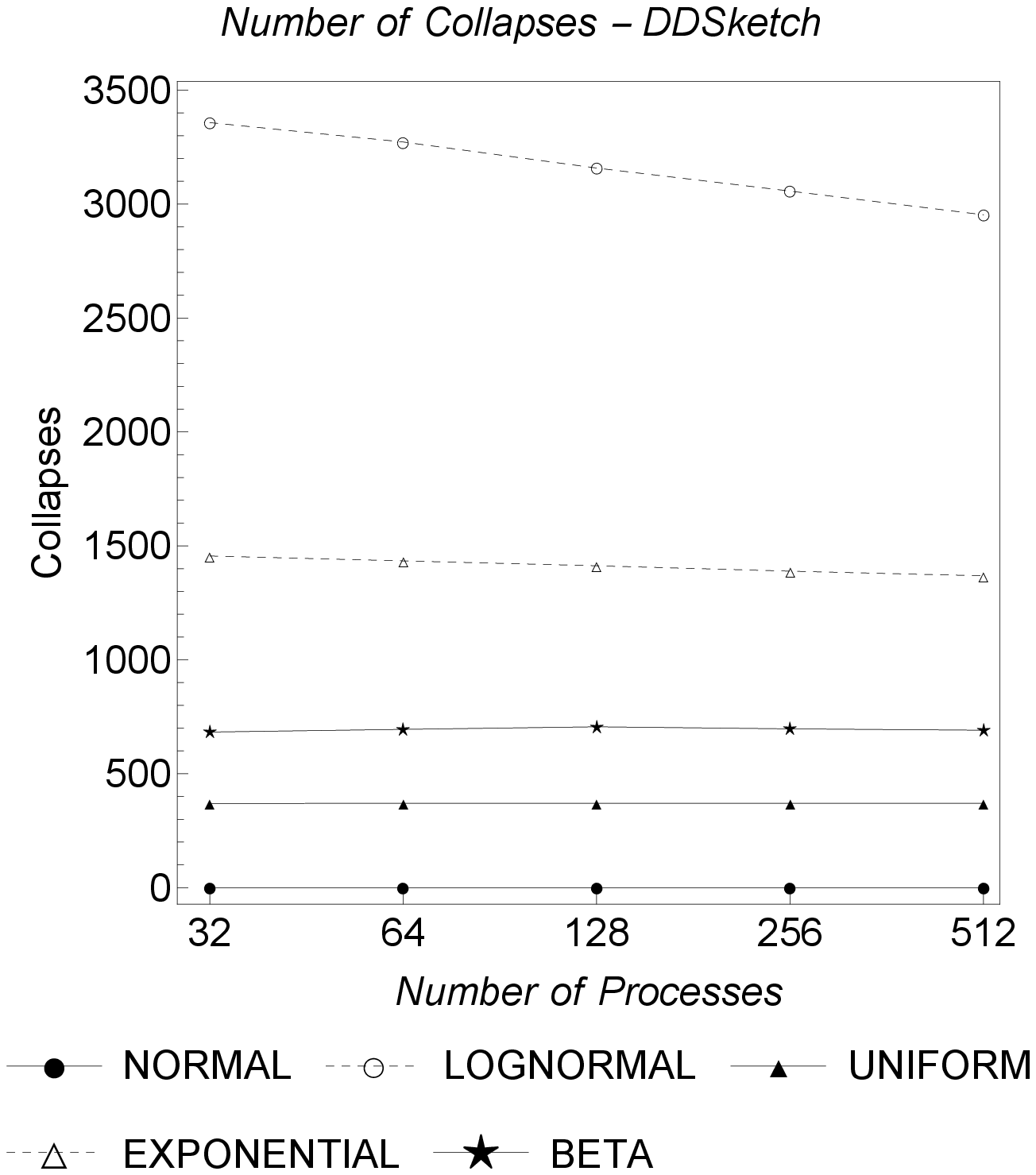}
		\label{collapses_ddog}
		} &
		
		\subfloat[]{
			\includegraphics[width=0.45\textwidth]{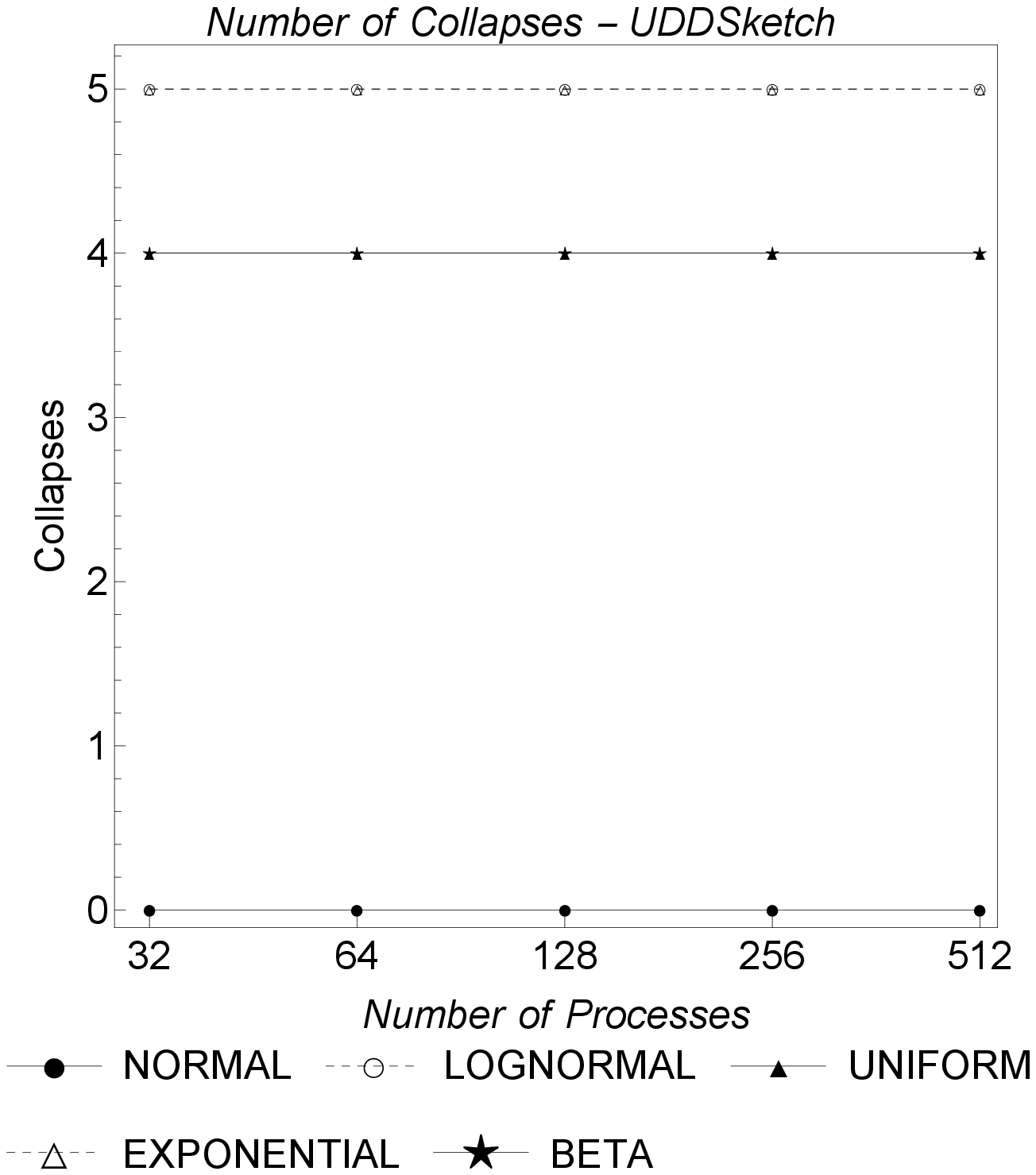}
			\label{collapses_udd}
		} 
	\end{tabular}
	
	\caption{Number of sketch collapses varying the input distribution.} 
	\label{collapses}
\end{figure*}

The parallel computation performance is shown in Fig. \ref{execution_time} in which we use log-log plots to represent the parallel running time of \textsc{DDSketch} and \textsc{UDDSketch} with different input distributions. The log-log plots give also a clear evidence of the parallel scalability of the algorithms: indeed, the ideal parallel speedup is represented by curve with slope equal to $-1$. The results clearly show that our \textsc{UDDSketch} algorithm provides a good parallel scalability and its parallel running time is equal to the \textsc{DDSketch};  only with the exponential distribution \textsc{UDDSketch} is slightly slower than the \textsc{DDSketch} (the difference in the execution time is less than $5\%$). 

Moreover, \textsc{UDDSketch} outperforms \textsc{DDSketch} with regard to the accuracy. Table \ref{accuracy} reports the $q_0$-accuracy and $\alpha$ value at the end of computation; as shown, \textsc{UDDSketch} has a $q_0$-accuracy equal to 0 for every distributions, which means that it can provide an accurate estimation for all of the quantiles, with a relative error less than $\alpha$; instead, \textsc{DDSketch} is accurate only for those quantiles greater than $q_0$ which, for some distributions like the exponential and the lognormal, is greater than $0.99$, demonstrating that the sketch size is not big enough to guarantee a quantile estimation with an error less than $\alpha=0.001$. The \textsc{UDDSketch} algorithm, instead, is self adaptive and consistently makes good use of the available space: for the exponential and the lognormal distributions it uses a greater value for $\alpha$ to guarantee a quantile estimation along all of the quantiles range with an error as small as possible using the sketch size defined by the user. Therefore, the results confirm that the parallel version of the \textsc{UDDSketch} algorithm outperforms \textsc{DDSketch} with regard to the accuracy, and that simultaneously it exhibits good parallel scalability and a running time comparable with \textsc{DDSketch}.    

\begin{figure*}[h]
	\centering
	\begin{tabular}{ccc}		
		\subfloat[]{
		\includegraphics[width=0.3\textwidth]{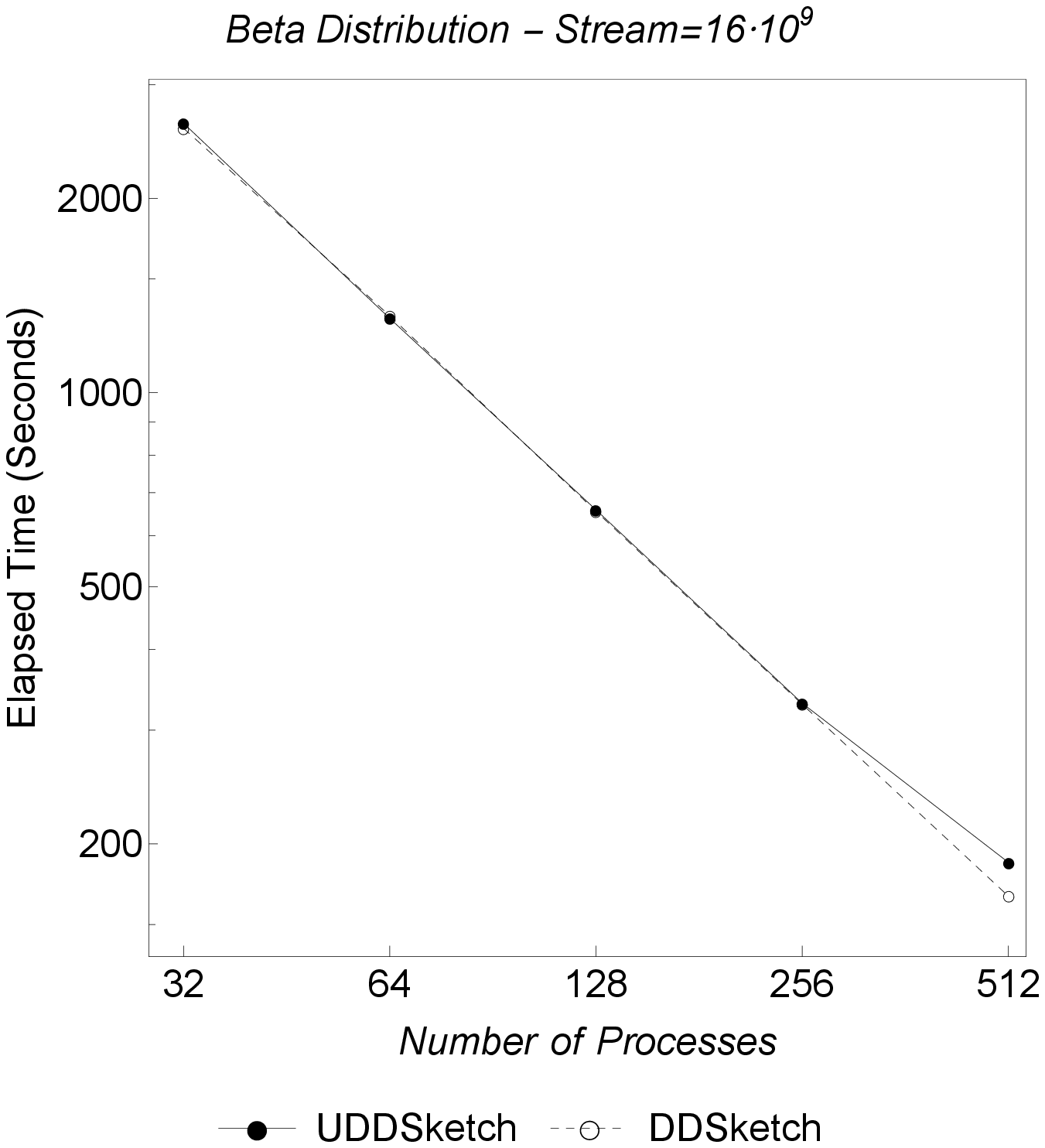}
		\label{beta}
		} &
		
		\subfloat[]{
			\includegraphics[width=0.3\textwidth]{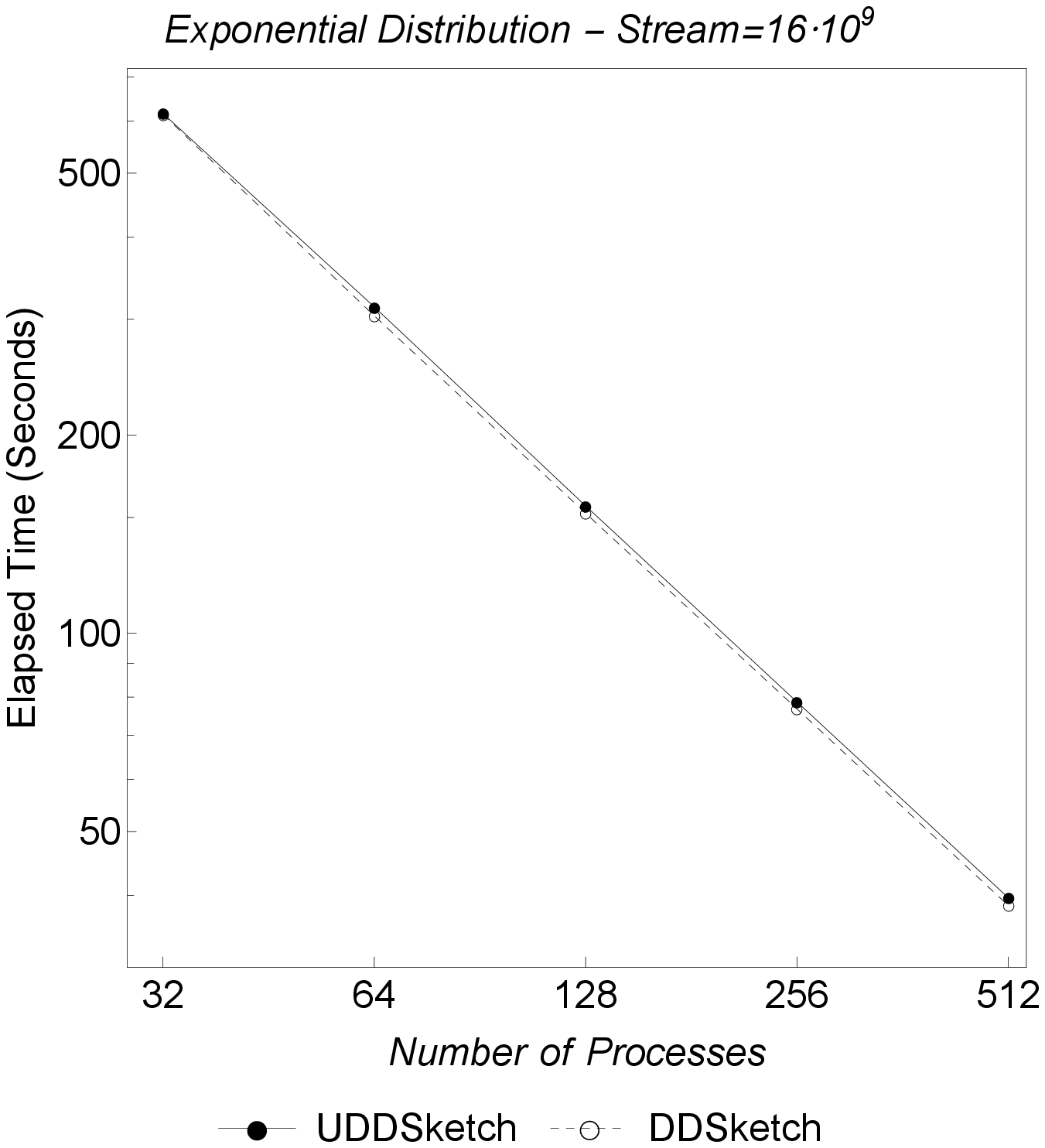}
			\label{exponential}
		} &
		\subfloat[]{
			\includegraphics[width=0.3\textwidth]{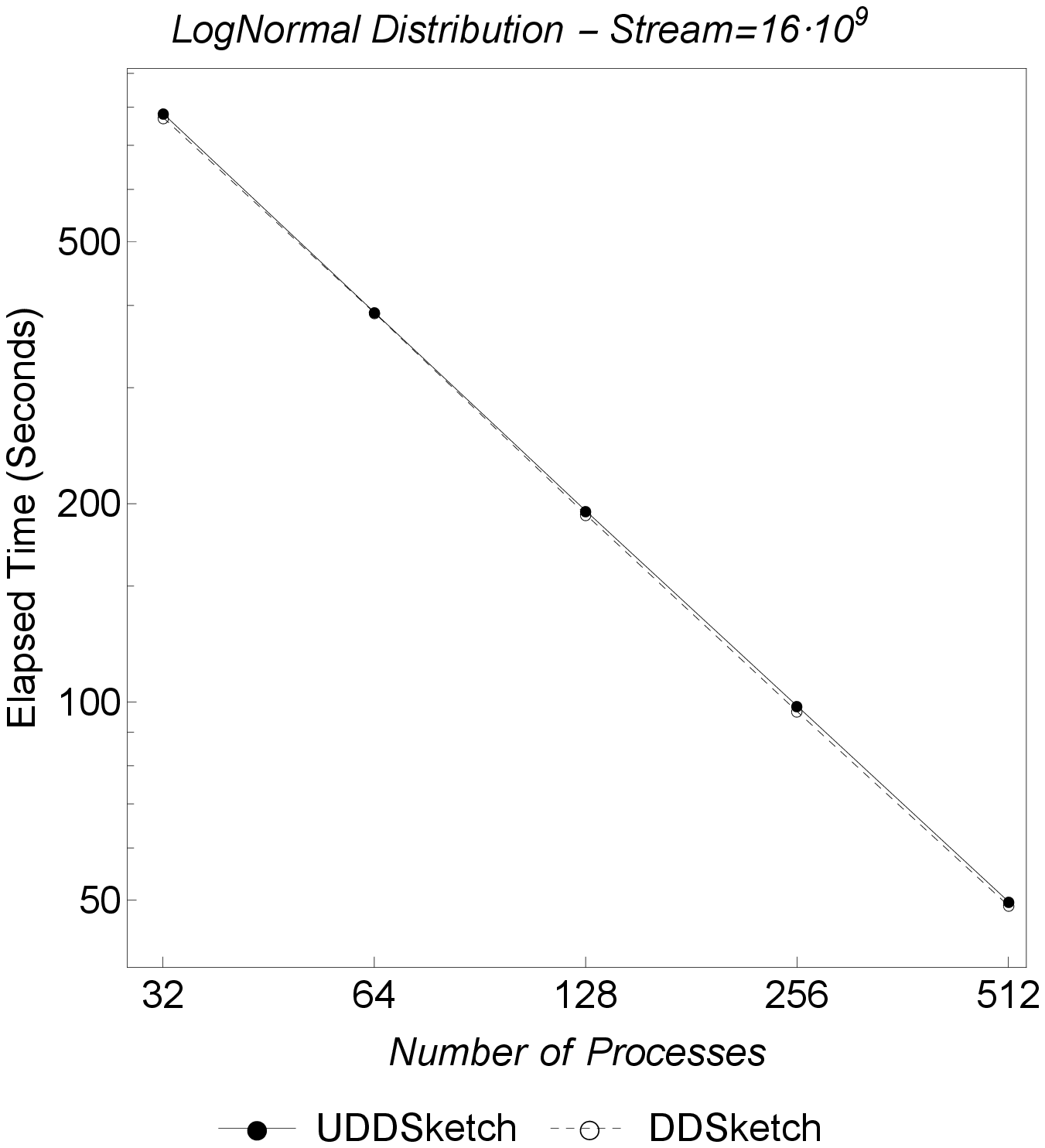}
			\label{lognormal}
		} \\
		
		\subfloat[]{
		\includegraphics[width=0.3\textwidth]{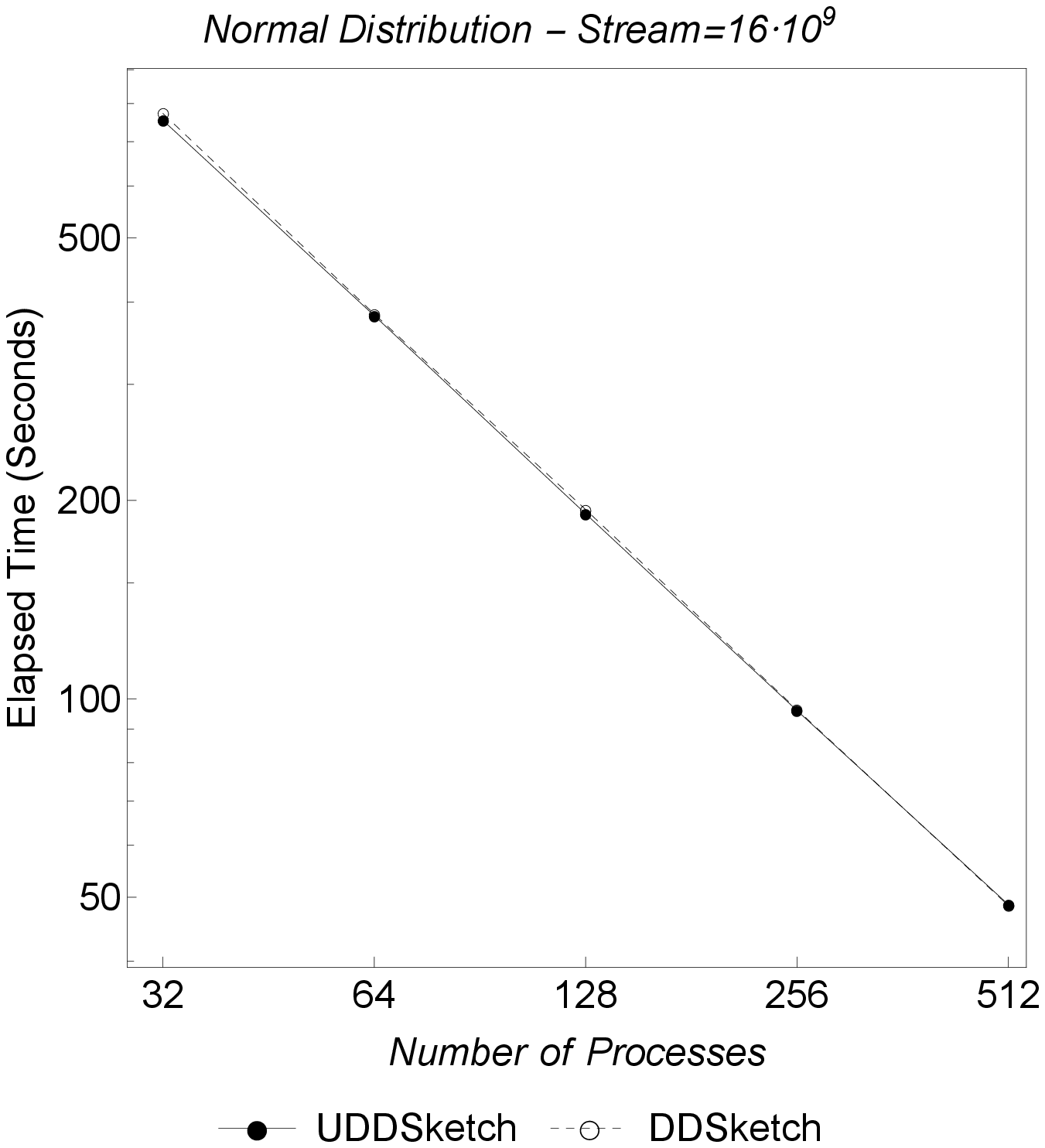}
		\label{normal}
		} &
		
		\subfloat[]{
			\includegraphics[width=0.3\textwidth]{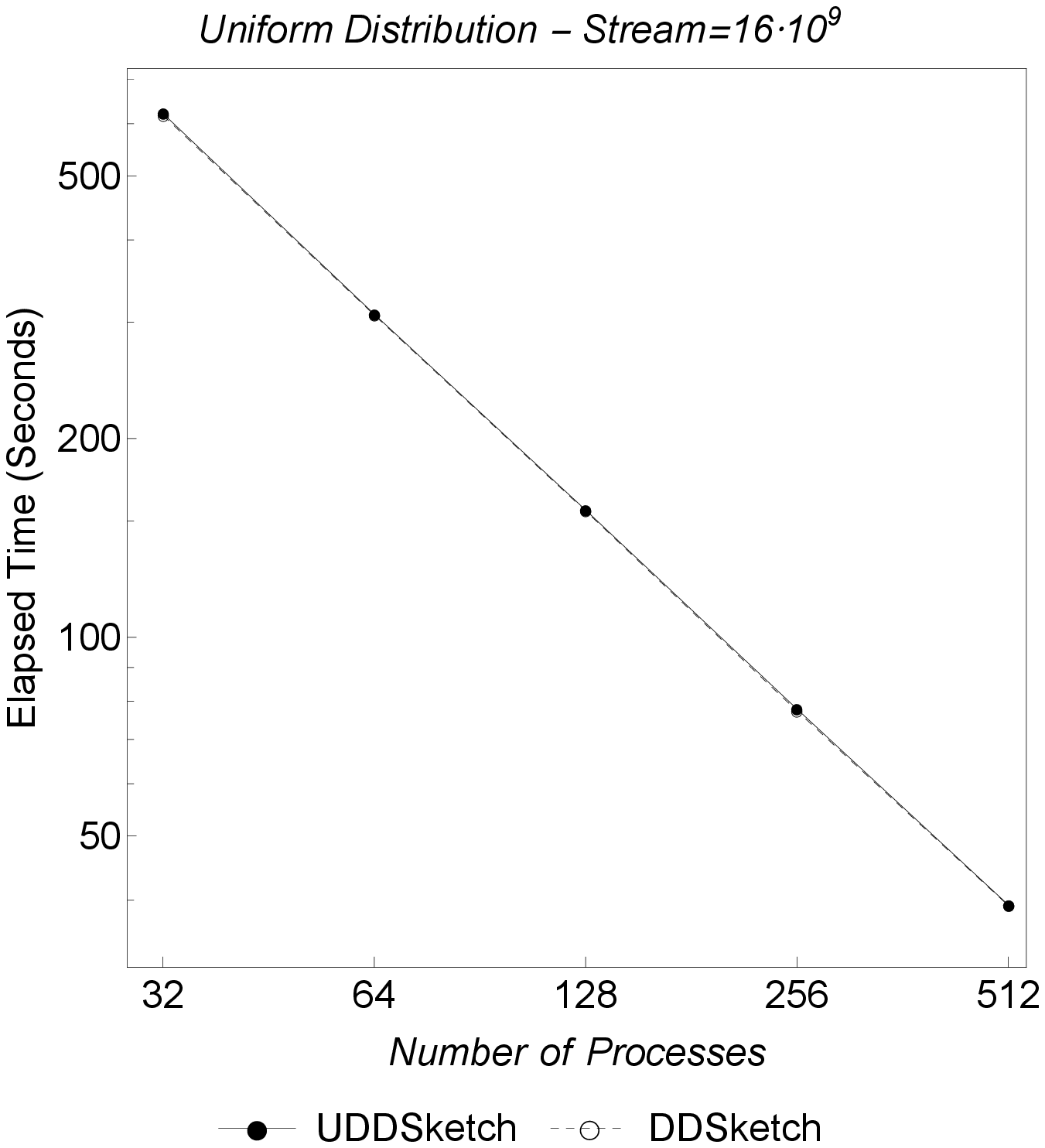}
			\label{uniform}
		}
	\end{tabular}
	
	\caption{Parallel Running time varying the number of processes in log-log plots.} 
	\label{execution_time}
\end{figure*}

\begin{table}
    \small
	\caption{Accuracy}
	\label{accuracy}
	\centering
	%\ra{1.4}
	\begin{tabular}{@{}l|cc|cc@{}}
		& \multicolumn{2}{c|}{\textsc{DDSketch}} & \multicolumn{2}{c}{\textsc{UDDSketch}}\\ 
		\textbf{Dataset} & \textbf{$q_0$-accuracy} & \textbf{$\alpha$} & \textbf{$q_0$-accuracy} & \textbf{$\alpha$}\\
		\hline
	beta & 0.798 & 0.001 & 0 & 0.019 \\ 
	exponential & 0.998 & 0.001 & 0 & 0.031 \\ 
	lognormal & 0.999 & 0.001 & 0 & 0.031 \\ 
	normal & 0 & 0.001 & 0 & 0.001 \\ 
	uniform & 0.360 & 0.001 & 0 & 0.016 \\ 
	\end{tabular}
\end{table}

\section{Conclusions}
\label{conclusions}
In this paper we have introduced a parallel version of the \textsc{UDDSketch} algorithm for accurate quantile tracking and analysis, suitable for message-passing based architectures. The algorithm allows compressing and fusing big volume data streams (or big data) retaining the error and size guarantees provided by the sequential \textsc{UDDSketch} algorithm. We have formally proved its correctness and compared it to a parallel version of \textsc{DDSketch}. The extensive experimental results confirm the validity of our approach, since our algorithm almost always outperforms the parallel \textsc{DDSketch} algorithm with regard to the overall accuracy in determining the quantiles, providing simultaneously a good parallel scalability.

\section*{Acknowledgments} The authors would like to thank CINECA for granting the access to the Marconi M100 supercomputer machine through grant IsC80\_PDQA HP10CZD477, and Euro Mediterranean Center on Climate Change, Foundation, Italy for granting the access to the Zeus supercomputer machine.

%%%%%%%%%%%%%%%%%%%%%%%%%%%%%%%%%%%%%%%%%%%%%%%%%%%%%%%%%%%
% the following \clearpage command will prevent floats to appear in or after the references
\clearpage

\bibliographystyle{elsarticle-num}
\bibliography{bibliography}

\end{document}